\RequirePackage{snapshot}

\documentclass[conference,twocolumn]{IEEEtran}%
\usepackage{times}
\usepackage{amssymb}
\usepackage{amsfonts}
\usepackage{amsmath}%
\usepackage{subfigure}
\ifCLASSINFOpdf 
\usepackage[pdftex]{graphicx} 
\usepackage[pdftex,colorlinks,%draft, %, %naturalnames,% 
           bookmarks=true,%
           ]{hyperref}  
\usepackage[numbers,sort&compress]{natbib} 
\usepackage{hypernat} 
\usepackage{epstopdf}
\else

\usepackage[dvips]{graphicx} 
\usepackage[dvips,colorlinks,%draft, %, %naturalnames,% 
           bookmarks=true,%
           ]{hyperref}  
 \usepackage[numbers,sort&compress]{natbib} 
\fi 
\hypersetup{
   bookmarksnumbered,
   pdfstartview={FitH},
   citecolor={blue},
   linkcolor={red},
   urlcolor=[rgb]{0,0.55,0},%{green},
   pdfpagemode={UseOutlines},
 pdftitle={What\ frequency\ bandwidth\ to\ run\ cellular\ network\ in\ a\ given\ country?\ ---\ a\ downlink\ dimensioning\ problem},% 
   pdfauthor={B.\ Blaszczyszyn\ -\ M.\ K.\  Karray}
} 

\graphicspath{{./}{./Figures/}{./BandwidthVsTraffic/}}

\newtheorem{acknowledgement}{Acknowledgement}

\newtheorem{cor}{Corollary}
\newenvironment{corollary}{\bf\begin{cor}\rm\em}{\end{cor}} % corollary 

\newtheorem{prop}{Proposition}
\newenvironment{proposition}{\bf\begin{prop}\rm\em}{\end{prop}} % corollary

\newenvironment{proof}[1][Proof]{\textbf{#1.} }{\ \rule{0.5em}{0.5em}}

\begin{document}

\title{\huge What frequency bandwidth to run cellular network in a given country? --- a downlink dimensioning problem}
\author{Bart{\l }omiej~B{\l }aszczyszyn$^\dagger$  and Mohamed~K.~Karray$^\ast$\\[-10ex]}

\maketitle

\begin{abstract}
We propose an analytic approach to the frequency bandwidth dimensioning problem, faced by cellular network operators who deploy/upgrade their networks in various geographical regions (countries) with an inhomogeneous urbanization. We present a model allowing one to capture fundamental relations between users' quality of service parameters (mean downlink throughput), traffic demand, the density of base station deployment, and the available frequency bandwidth. These relations depend on the applied cellular technology (3G or 4G impacting user peak bit-rate) and on the path-loss characteristics observed in different (urban, sub-urban and rural) areas. We observe that if the distance between base stations is kept inversely proportional to the distance coefficient of the path-loss function, then the performance of the typical cells of these different areas is similar when serving the same (per-cell) traffic demand. In this case, the frequency bandwidth dimensioning problem can be solved uniformly across the country applying the mean cell approach proposed in~\cite{BlaszczyszynJK2014TypicalCell}. We validate our approach by comparing the analytical results to measurements in operational networks in various geographical zones of different countries.

\end{abstract}

%\title{Performance evaluation and bandwidth dimensioning of wireless
%cellular networks at a country scale}


\let\thefootnote\relax\footnotetext{\hspace{-2ex}$^\dagger$Inria/Ens, 23
  av. d'Italie 75214 Paris, France; Bartek.Blaszczyszyn@ens.fr\\ 
$^\ast$Orange Labs, 38/40 rue G\'{e}n\'{e}ral Leclerc, 92794
  Issy-Moulineaux, France; mohamed.karray@orange.com}
\newcommand{\thefootnote}{\arabic{footnote}}

\section{Introduction}

A systematic increase of the traffic in wireless cellular networks leads to a
potential degradation of the users'~quality of service (QoS). In order to
prevent such degradation, network operators (besides upgrading the technology)
keep adding new base stations and/or supplementary frequency bandwidth.
Planning of this network dimensioning process requires the knowledge of the
relations between the QoS parameters, traffic demand, density of base station
deployment, and the operated frequency bandwidth. Such relations are usually
developed for homogeneous network models representing, separately, a typical
city or a rural area. To the best of our knowledge, a global approach to some
larger geographical zone, as e.g. some given country, typically with an
inhomogeneous urbanization, has not yet been proposed. This is the main goal
of this paper.

Our work is based on~\cite{BlaszczyszynJK2014TypicalCell}, where a
typical-cell approach was proposed to the aforementioned study of the QoS,
with the information-theoretic characterization of the peak-bit rate, the
queueing-theoretic evaluation of the cell performance, and a
stochastic-geometric approach to irregular but homogeneous network deployment.
In this paper we extend this approach to an inhomogeneous network deployment,
i.e., a network which is supposed to cover urban, sub-urban and rural areas,
for which path-loss characteristics and the network deployment densities are different.

A key observation that we make in this regard is that if the distance between
neighbouring base stations is kept inversely proportional (equivalently, the
square root of the network density is proportional) to the distance
coefficient of the path-loss function, then the typical cells of these
different areas exhibit the same \textquotedblleft response\textquotedblright%
\ to the traffic demand. Specifically, they exhibit the same relation between
the mean user throughput and the per-cell traffic demand. Thus, one can say
that this inhomogeneous network offers to its users homogeneous
\textquotedblleft QoS response\textquotedblright\ across all different areas.
More precisely: cells in rural areas will be larger than these in urban areas,
and the traffic demand per surface can be different, with urban areas
generally serving more traffic. However, in the case when the mean traffic
demand per cell is equal, these areas will \textquotedblleft
offer\textquotedblright\ the same user's throughput.

As a consequence, the bandwidth dimensioning problem can be addressed and
solved for one particular, say urban, area and applied to the whole network,
appropriately adapting the density of base station in sub-urban and rural
areas to the specific value of the distance coefficient in the path-loss function.

We validate our approach by comparing the obtained results to measurements in
3G and 4G operational networks in different countries in Europe, Africa and Asia.

\subsection{Related works}

The evaluation of the performance of cellular networks is a complex problem,
but crucial for network operators. It motivates a lot of engineering and
research studies. The complexity of this problem made many actors develop
complex and time consuming simulation tools such as those developed by the
industrial contributors to 3GPP (\emph{3rd Generation Partnership
Project})~\cite{3GPP36814-900}. There are many other simulation tools such as
TelematicsLab LTE-Sim~\cite{Piro2011}, University of Vien LTE
simulator~\cite{Mehlfuhrer2011,Simko2012} and LENA
tool~\cite{Baldo2011,Baldo2012} of CTTC, which are not necessarily compliant
with 3GPP.

A possible analytical approach to this problem relies on the information
theoretic characterization of the individual link performance; cf
e.g.~\cite{GoldsmithChua1997,Mogensen2007}, in conjunction with a queueing
theoretic modeling and analysis of the user traffic; cf.
e.g.~\cite{Borst2003,BonaldProutiere2003,HegdeAltman2003,BonaldBorstHegdeJP2009,RongElayoubiHaddada2011,KarrayJovanovic2013Load}%
. Recently,~\cite{BlaszczyszynJK2014TypicalCell,hetnets} proposes an approach
combining queueing theory and stochastic geometry allowing to study the
dependence of the mean user throughput on the traffic demand in a typical
city. This analytical approach is validated with field measurements at the
scale of a city. At this scale, not only the network average metrics but also
their spatial distributions in the considered, homogeneous, zones are well
predicted in~\cite{BlaszczyszynJK2014QoS_Distribution}. This prior work does
not consider inhomogeneous networks, observed at the scale of a larger
geographical zone or a whole country. In particular, the simulation tools are
too time consuming to be applicable at this scale.

\subsection{Paper organization}

We describe our general cellular network model, comprising the geometry of
base stations, propagation loss model, traffic dynamics, and service policy in
Section~\ref{s.ModelDescription}. In Section~\ref{s.ModelAnalysis} we propose
a model permitting to get the relation of users'~QoS to the key network
parameters at a country scale. This approach is validated by comparing, in
Section~\ref{s.NumericalResults}, the obtained results to measurements in
various operational networks in different countries. We show also how to solve
the bandwidth dimensioning problem.

\section{Homogeneous model description}

\label{s.ModelDescription} In this section we briefly recall the model
from~\cite{BlaszczyszynJK2014TypicalCell}.

\subsection{Network geometry and propagation}

The network is composed of base stations (BS) whose locations are modelled by
a stationary point process $\Phi=\left\{  X_{n}\right\}  _{n\in\mathbb{Z}}$ on
$\mathbb{R}^{2}$ with intensity $\lambda>0$. We assume that $\Phi$\ is
stationary and ergodic. Each base station $X_{n}$\ emits a power denoted by
$P_{n}$. We assume that $\left\{  P_{n}\right\}  _{n\in\mathbb{Z}}$\ are
i.i.d. marks of $\Phi$.

The propagation loss comprises a deterministic effect depending on the
distance between the transmitter and receiver called \emph{path-loss}, and a
random effect called \emph{shadowing}. The path-loss is modeled by a function
\begin{equation}
l\left(  x\right)  =\left(  K\left\vert x\right\vert \right)  ^{\beta},\quad
x\in\mathbb{R}^{2}, \label{e.DistanceLoss}%
\end{equation}
where $K>0$\ and $\beta>2$\ are two parameters called \emph{distance
coefficient} and \emph{exponent} respectively.\ The shadowing between BS
$X_{n}$\ and all the locations $y\in\mathbb{R}^{2}$\ is modelled by a
stochastic process $S_{n}\left(  y-X_{n}\right)  $. We assume that the
processes $\left\{  S_{n}\left(  \cdot\right)  \right\}  _{n\in\mathbb{Z}}%
$\ are i.i.d. marks of $\Phi$.

The inverse of the power received at location $y$ from BS $X_{n}$ is denoted
by
\[
L_{X_{n}}\left(  y\right)  =\frac{l\left(
y-X_{n}\right)  }{P_{n}S_{n}\left(  y-X_{n}\right)  },
\quad y\in\mathbb{R}^{2},n\in\mathbb{Z}%
\]
and will be called, with a slight abuse of terminology, \emph{propagation
loss} between $y$ and $X_{n}$.

In order to simplify the notation, we shall omit the index $n$ of BS $X_{n}$.
Each BS $X\in\Phi$ serves the locations where the received power is the
strongest among all the BS; that is
\begin{equation}
V\left(  X\right)  =\left\{  y\in\mathbb{R}^{2}:L_{X}\left(  y\right)  \leq
L_{Y}\left(  y\right)  \text{ for all }Y\in\Phi\right\}  \label{e.Cell}%
\end{equation}
called \emph{cell} of $X$.

The signal-to-interference-and-noise (SINR) power ratio in the \emph{downlink}
for a user located at $y\in V\left(  X\right)  $ equals
\begin{equation}
\mathrm{SINR}\left(  y,\Phi\right)  =\frac{1/L_{X}\left(  y\right)  }%
{N+\sum_{Y\in\Phi\backslash\left\{  X\right\}  }\varphi_{Y}/L_{Y}\left(
y\right)  }, \label{e.SINR}%
\end{equation}
where $N$ is the noise power and each $\varphi_{Y}\in\left[  0,1\right]  $\ is
some \emph{interference factor} accounting for the activity of BS\ $Y$ in a
way that will be made specific in Section~\ref{s.CellCharacteristics}. We
assume that $\left\{  \varphi_{Y}\right\}  _{Y\in\Phi}$ are additional (not
necessarily independent) marks of the point process $\Phi$.

A single user served by BS $X$ and located at $y\in V(X)$\ gets a bit-rate
$R\left(  \mathrm{SINR}\left(  y,\Phi\right)  \right)  $, called \emph{peak
bit-rate}, which is some function of the SINR given by~(\ref{e.SINR}).
Particular form of this function depends on the actual technology used to
support the wireless link.

\subsection{Traffic dynamics}

We consider variable bit-rate (VBR) traffic; i.e., users arrive to the network
and require to transmit some volume of data at a bit-rate decided by the
network. Each user arrives at a location uniformly distributed and requires to
transmit a random volume of data of mean $1/\mu$. The duration between the
arrivals of two successive users in each zone of surface $S$ is an exponential
random variable of parameter $\gamma\times S$ (i.e. there are $\gamma
$\ arrivals per surface unit). The arrival locations, inter-arrival durations
as well as the data volumes are assumed independent of each other. We assume
that the users don't move during their calls. The traffic demand \emph{per
surface unit} is then equal to
$\rho={\gamma}/{\mu}$,
which may be expressed in bit/s/km$^{2}$.

The traffic demand \emph{to a given cell} equals
\begin{equation}
\rho\left(  X\right)  =\rho\left\vert V\left(  X\right)  \right\vert ,\quad
X\in\Phi,\label{e.TrafficDemand}%
\end{equation}
where $\left\vert A\right\vert $\ is the surface of $A$.

We shall assume that each user in a cell gets an equal portion of time for his
service. Thus when there are $k$ users in a cell, each one gets a bit-rate
equal to his peak bit-rate divided by $k$. More explicitly, if a base station
located at $X$ serves $k$ users located at $y_{1},y_{2},\ldots,y_{k}\in
V\left(  X\right)  $\ then the bit-rate of the user located at $y_{j}$\ equals
$\frac{1}{k}R\left(  \mathrm{SINR}\left(  y_{j},\Phi\right)  \right)  $,
$j\in\left\{  1,2,\ldots,k\right\}  $.

\subsection{Cell performance metrics}

\label{s.CellCharacteristics}We consider now the stationary state of the
network in the long run of the call arrivals and departures. Using queuing
theory tools, it is proven in~\cite[Proposition~1]{KarrayJovanovic2013Load}\ that:

\begin{itemize}
\item Each base station $X\in\Phi$\ can serve the traffic demand within its
cell if this latter doesn't exceed some \emph{critical} value which is the
harmonic mean of the peak bit-rate over the cell; that is
\begin{equation}
\rho_{\mathrm{c}}\left(  X\right)  :=\left\vert V\left(  X\right)  \right\vert
\left[  \int_{V\left(  X\right)  }1/R\left(  \mathrm{SINR}\left(
y,\Phi\right)  \right)  \mathrm{d}y\right]  ^{-1}, \label{e.CriticalTraffic}%
\end{equation}

\item The mean user throughput in cell $V\left(  X\right)  $ equals
\begin{equation}
r\left(  X\right)  =\max(\rho_{\mathrm{c}}\left(  X\right)  -\rho\left(
X\right)  ,0). \label{e.UserThroughput}%
\end{equation}

\item The mean number of users in cell $V\left(  X\right)  $ equals
\begin{equation}
N\left(  X\right)  =\frac{\rho\left(  X\right)  }{r\left(  X\right)  }.
\label{e.UsersNumber}%
\end{equation}

\item Moreover, we define the \emph{cell load} as
\begin{equation}
\theta\left(  X\right)  =\frac{\rho\left(  X\right)  }{\rho_{\mathrm{c}%
}\left(  X\right)  }=\rho\int_{V\left(  X\right)  }1/R\left(  \mathrm{SINR}%
\left(  y,\Phi\right)  \right)  \mathrm{d}y. \label{e.Load1}%
\end{equation}

\item The probability that the base station has at least one user to serve (at
a given time) equals
\begin{equation}
p\left(  X\right)  =\min\left(  \theta\left(  X\right)  ,1\right).
\label{e.Proba}%
\end{equation}

\end{itemize}

We assume that a BS transmits only when it serves at least one user. Then, as
proposed in~\cite{BlaszczyszynJK2014TypicalCell}, we take $\varphi_{Y}=p(Y)$
in the SINR expression~(\ref{e.SINR}). Thus~(\ref{e.Load1}) becomes
\begin{equation}
\theta\left(  X\right)  =\rho\int_{V\left(  X\right)  }1/R\left(
\frac{1/L_{X}\left(  y\right)  }{N+\sum_{Y\in\Phi\backslash\left\{  X\right\}
}\frac{\min\left(  \theta\left(  Y\right)  ,1\right)  }{L_{Y}\left(  y\right)
}}\right)  \mathrm{d}y, \label{e.LoadEquations}%
\end{equation}
which is a system of equations with unknown cell loads $\left\{  \theta\left(
X\right)  \right\}  _{X\in\Phi}$.

\subsection{Network performance metrics}

\label{s.MeanCell}The network performance metrics are defined by averaging
spatially over all the cells in the network in an appropriate
way~\cite{BlaszczyszynJK2014TypicalCell}. In particular, it follows from the
ergodic theorem for point processes~\cite[Theorem~13.4.III]%
{DaleyVereJones2003}\ that the average cell load equals
\[
\lim_{|A|\rightarrow\infty}\frac{\sum_{X\in\Phi\cap A}\theta(X)}{\Phi
(A)}=\mathbf{E}^{0}\left[  \theta\left(  0\right)  \right]  =\frac{\rho
}{\lambda}\mathbf{E}\left[  \frac{1}{R\left(  \mathrm{SINR}\left(
0,\Phi\right)  \right)  }\right],
\]
where $A$ is a ball centered at the origin having radius increasing to
infinity and $\mathbf{E}^{0}$\ is the expectation with respect to the Palm
probability associated to $\Phi$. The second equality in the above equation is
proved in~\cite[Proposition~3]{BlaszczyszynJK2014TypicalCell}.

Moreover, it is observed in~\cite{BlaszczyszynJK2014TypicalCell} that the
network performance metrics are well approximated by the \emph{mean cell}
model. Specifically, we define a virtual cell having traffic demand $\bar
{\rho}:=\mathbf{E}^{0}\left[  \rho\left(  0\right)  \right]  $ and load
$\bar{\theta}:=\mathbf{E}^{0}\left[  \theta\left(  0\right)  \right]  $. For
other performance metrics, the mean cell mimics the behavior of the true cells
given in Section~\ref{s.CellCharacteristics}; that is it has critical traffic
demand deduced from~(\ref{e.Load1})
\[
\bar{\rho}_{\mathrm{c}}:=\frac{\bar{\rho}}{\bar{\theta}},%
\]
mean user's throughput deduced from~(\ref{e.UserThroughput})
\[
\bar{r}:=\max\left(  \bar{\rho}_{\mathrm{c}}-\bar{\rho},0\right),
\]
and the mean number of users deduced from~(\ref{e.UsersNumber})
\[
\bar{N}:=\frac{\bar{\rho}}{\bar{r}}.%
\]

The load equations~(\ref{e.LoadEquations}) become for the mean cell
\[
\bar{\theta}=\frac{\rho}{\lambda_{\mathrm{BS}}}\mathbf{E}\left[  1/R\left(
\frac{1/L_{X^{\ast}}\left(  0\right)  }{N+\bar{\theta}\sum_{Y\in\Phi
\backslash\left\{  X\right\}  }1/L_{Y}\left(  0\right)  }\right)  \right],
\]
where $X^{\ast}$\ is the location of the BS whose cell covers the origin.

\section{From scaling equations to inhomogeneous networks}

\label{s.ModelAnalysis} We begin by studying some scaling laws observed in
homogeneous networks. Then we will introduce some inhomogeneous networks which
are able to offer homogeneous QoS response to the traffic demand.

\subsection{Scaling laws for homogeneous networks}

\label{ss.Scaling} Consider a homogeneous network model descried in
Section~\ref{s.ModelDescription}. For $\alpha>0$ consider a network obtained
from this original one
%(introduced in Section~\ref{s.ModelDescription})
by scaling the base station locations $\Phi^{\prime}=\left\{  X^{\prime
}=\alpha X\right\}  _{X\in\Phi}$, the traffic demand intensity $\rho^{\prime
}=\rho/\alpha^{2}$, distance coefficient $K^{\prime}=K/\alpha$ and shadowing
processes $S_{n}^{\prime}\left(  y\right)  =S_{n}\left(  \frac{y}{\alpha
}\right)  $, while preserving the original (arbitrary) marks (interference
factors) $\varphi_{X^{\prime}}^{\prime}=\varphi_{X}$, $X\in\Phi$ and powers
$P_{n}^{\prime}=P_{n}$. For the rescaled network consider the cells
$V^{\prime}(X^{\prime})$ given by~(\ref{e.Cell}) and their characteristics
$\rho^{\prime}(X^{\prime})$, $\rho_{\mathrm{c}}^{\prime}(X^{\prime})$,
$r^{\prime}(X^{\prime})$, $N^{\prime}(X^{\prime})$, $\theta^{\prime}%
(X^{\prime})$, $p^{\prime}(X^{\prime})$ calculated as
in~(\ref{e.TrafficDemand}), (\ref{e.CriticalTraffic}), (\ref{e.UserThroughput}%
), (\ref{e.UsersNumber}), (\ref{e.Load1}) and~(\ref{e.Proba}), respectively.

\begin{proposition}
\label{p.Scaling} For any $X^{\prime}\in\Phi^{\prime}$, we have
%\begin{equation}
%V^{\prime}\left(  \alpha X_{n}\right)  =\alpha V\left(  X_{n}\right)
%\label{e.ScalingArea}%
%\end{equation}
$V^{\prime}\left(  \alpha X_{n}\right)  =\alpha V\left(  X_{n}\right)  $ while
$\rho^{\prime}(X^{\prime})=\rho(X)$, $\rho_{\mathrm{c}}^{\prime}(X^{\prime
})=\rho_{\mathrm{c}}(X)$, $r^{\prime}(X^{\prime})=r(X)$, $N^{\prime}%
(X^{\prime})=N(X)$, $\theta^{\prime}(X^{\prime})=\theta(X)$.
\end{proposition}

\begin{proof}
Observe that
\[
L_{\alpha X_{n}}^{\prime}\left(  y\right) \!=\!\frac{P_{n}S_{n}^{\prime}\left(
y-\alpha X_{n}\right)  }{\left(  K^{\prime}\left\vert y-\alpha X_{n}%
\right\vert \right)  ^{\beta}}\!=\!\frac{P_{n}S_{n}\left(  \frac{y}{\alpha}%
-X_{n}\right)  }{\left(  K\left\vert \frac{y}{\alpha}-X_{n}\right\vert
\right)  ^{\beta}}\!=\!L_{X_{n}}{\textstyle\left(  \frac{y}{\alpha}\right)  }%
\]
and by~(\ref{e.Cell}) one gets the required dilation relation for the
individual cells. Moreover, using~(\ref{e.SINR}) one gets by simple algebra
%\begin{equation}
%\mathrm{SINR}\left(  y,\Phi^{\prime}\right)  =\mathrm{SINR}\left(
%y/\alpha,\Phi\right)  \label{e.ScalingSINR}%
%\end{equation}
$\mathrm{SINR}\left(  y,\Phi^{\prime}\right)  =\mathrm{SINR}\left(
y/\alpha,\Phi\right)  $. Starting from~(\ref{e.CriticalTraffic}) and making
the change of variable $z=y/\alpha$, it follows that
\begin{align*}
\rho_{\mathrm{c}}^{\prime}\left(  X^{\prime}\right)   &  =\left\vert
V^{\prime}\left(  X^{\prime}\right)  \right\vert \left[  \int_{V^{\prime
}\left(  X^{\prime}\right)  }1/R\left(  \mathrm{SINR}\left(  y,\Phi^{\prime
}\right)  \right)  \mathrm{d}y\right]  ^{-1}\\
&  =\alpha^{2}\left\vert V\left(  X\right)  \right\vert \left[  \int_{\alpha
V\left(  X\right)  }1/R\left(  \mathrm{SINR}\left(  y/\alpha,\Phi\right)
\right)  \mathrm{d}y\right]  ^{-1}\\
&  =\rho_{\mathrm{c}}\left(  X\right).
\end{align*}
%where for the third equality we use~(\ref{e.ScalingArea}) and for the fourth
%one we use~(\ref{e.ScalingSINR}).
The remaining desired equalities then follow from the fact that
\[
\rho^{\prime}\left(  X^{\prime}\right)
%=\rho^{\prime}\left\vert V^{\prime}\left(  \alpha X\right)  \right\vert
=\frac{\rho}{\alpha^{2}}\left\vert \alpha V\left(  X\right)  \right\vert
=\rho\left\vert V\left(  X\right)  \right\vert =\rho\left(  X\right).
\]
\end{proof}

\begin{corollary}
Consider the rescaled network as in Proposition~\ref{p.Scaling} with the
interference factors taken $\varphi_{X}=\min\left(  \theta\left(  X\right)
,1\right)  $. Then the load equations~(\ref{e.LoadEquations}) are the same for
the two networks $\Phi$\ and $\Phi^{\prime}$. Therefore, the load factors
solving these equation are the same $\theta^{\prime}\left(  X^{\prime}\right)
=\theta\left(  X\right)  $ and consequently, $\rho_{\mathrm{c}}^{\prime
}(X^{\prime})=\rho_{\mathrm{c}}(X)$, $r^{\prime}(X^{\prime})=r(X)$,
$N^{\prime}(X^{\prime})=N(X)$, $\theta^{\prime}(X^{\prime})=\theta(X)$.
\end{corollary}

From the above observations we can deduce now that the considered scaling of
the network parameters preserves the mean cell characteristics defined in
Section~\ref{s.MeanCell}. Denote by $\mathbf{E}^{\prime}{}^{0}$\ the
expectation with respect to the Palm probability associated to $\Phi^{\prime}$.

\begin{corollary}
\label{c.Scaling} Consider the rescaled network as in
Proposition~\ref{p.Scaling} with arbitrary interference factors, possibly
satisfying the load equations~(\ref{e.LoadEquations}). Then $\mathbf{E}%
^{\prime0}\left[  \rho^{\prime}\left(  0\right)  \right]  =\mathbf{E}%
^{0}\left[  \rho\left(  0\right)  \right]  $ and $\mathbf{E}^{\prime0}\left[
\theta^{\prime}\left(  0\right)  \right]  =\mathbf{E}^{0}\left[  \theta\left(
0\right)  \right]  $. Consequently, the mean cells characteristics associated
to $\Phi$ and $\Phi^{\prime}$\ as described in Section~\ref{s.MeanCell} are identical.
\end{corollary}

\begin{proof}
It follows from the inverse formula of Palm calculus~\cite[Theorem~4.2.1]%
{BaccelliBlaszczyszyn2009T1} that $\mathbf{E}^{0}\left[  \rho\left(  0\right)
\right]  =\frac{\rho}{\lambda}$. Similarly, $\mathbf{E}^{\prime0}\left[
\rho^{\prime}\left(  0\right)  \right]  =\frac{\rho^{\prime}}{\lambda^{\prime
}}=\frac{\rho}{\lambda}$ which proves the first desired equality. The second
equality follows e.g. by the ergodic argument
\begin{align*}
\mathbf{E}^{\prime0}\left[  \theta^{\prime}\left(  0\right)  \right]   &
=\lim_{|A|\rightarrow\infty}\frac{1}{\Phi^{\prime}(A)}\sum_{X^{\prime}\in
\Phi^{\prime}\cap A}\theta^{\prime}(X^{\prime})\\
&  =\lim_{|A|\rightarrow\infty}\frac{1}{\Phi(A/\alpha)}\sum_{X\in\Phi\cap
A/\alpha}\theta(X)\\
&  =\lim_{|A|\rightarrow\infty}\frac{1}{\Phi(A)}\sum_{X\in\Phi\cap A}%
\theta(X)=\mathbf{E}^{0}\left[  \theta\left(  0\right)  \right].
\end{align*}
where the second equality is due to Proposition~\ref{p.Scaling}.
\end{proof}

\subsection{Inhomogeneous networks with homogeneous QoS response to the
traffic demand}

\label{ss.KversusDelta} Consider a geographic region (say a country), which is
composed of urban, suburban and rural areas. The parameters $K$ and $\beta
$\ of the distance loss model~(\ref{e.DistanceLoss}) depend on the type of the
zone. For example the COST-Hata model~\cite{Cost231_1999} gives the distance
loss function in the form $10\log_{10}\left(  l\left(  x\right)  \right)
=A+B\log_{10}\left(  \left\vert x\right\vert \right)  $ where $A$ and $B$\ are
given in Table~\ref{ta.LossParameters} and the distance $\left\vert
x\right\vert $ is in km. The corresponding value of the parameter $K=10^{A/B}%
$\ as well as the ratio $K_{\text{urban}}/K$ are also given in this table.

The density of base stations also depends on the type of the zone: it is
usually much higher in urban than in rural areas. Indeed, the distance $D$
between two neighboring base stations in rural areas may be up to $10$ times
larger than in dense urbane zones (where a typical value of $D$ is about $1$km).

Consider now a situation where the average distance between neighbouring base
stations is kept inversely proportional to the distance coefficient of the
path-loss function: $D\times K=\mathrm{const}$, or, in other words,
\begin{equation}
K_{i}/\sqrt{\lambda}_{i}=\mathrm{const}, \label{e.fractal}%
\end{equation}
where $\lambda_{i}$ is the network density in the given zone (assumed
homogeneous) and $K_{i}$ is the path-loss distance factor observed in this
zone. In our example shown in Table~\ref{ta.LossParameters} we should thus
have $D=1$,$~5$ and $8$\ km respectively for urban, suburban and rural
zones.~\footnote{Indeed, when operators deploys networks, they firstly aim to
assure some coverage condition which has the form $\left(  D\times K\right)
^{\beta}=\mathrm{const}$. Moreover, in urban zones, networks have to be
densified not only for capacity constraints, but also to assure coverage for
indoor users.}

Then the scaling laws proved in Section~\ref{ss.Scaling} say that locally, for
each homogeneous area of this inhomogeneous network, one will observe the same
relation between the mean user throughput and the (per-cell) traffic demand.
Cells in rural areas will be larger than these in urban areas, however in the
case when they are charged in the same way by the traffic demand (which, per
unit surface, needs to be correspondingly larger in urban areas) then they
offer the same QoS. In other words, one relation is enough to capture the key
dependence between the QoS, network density, bandwidth and the traffic demand
for different areas of this network. It can be used for the network
dimensioning.~\footnote{The approximation of an inhomogeneous network by a
piecewise homogeneous network can be made more precise e.g. following the
ideas presented
in~\cite{BlaszczyszynS2003Approximate,BlaszczyszynS2005Approximations}. Note
also that in our approach we ignored the fact that the path-loss exponent
$\beta$ is not the same for urban and the two other zones considered in
Table~\ref{ta.LossParameters}. This is clearly an approximation. Its quality can be studied 
using a network equivalence approach~\cite{BlaszczyszynK2013Equivalence} allowing for different path-loss
exponents.}%

%TCIMACRO{\TeXButton{B}{\begin{table}[tbp] \centering}}%
%BeginExpansion
\begin{table}[tbp] \centering
%EndExpansion%
\begin{tabular}
[c]{|l|l|l|l|l|}\hline
\textbf{Environment} & $A$ & $B$ & $K=10^{A/B}$ & $K_{\text{urban}}/K$\\\hline
\textbf{Urban} & $133.1$ & $33.8$ & $8667$ & $1$\\\hline
\textbf{Suburban} & $102.0$ & $31.8$ & $1612$ & $5$\\\hline
\textbf{Rural} & $97.0$ & $31.8$ & $1123$ & $8$\\\hline
\end{tabular}
\caption{Propagation parameters for carrier frequency 1795MHz from~\cite[Table 6.4]{Lagrange1995}.}\label{ta.LossParameters}%
%TCIMACRO{\TeXButton{E}{\end{table}}}%
%BeginExpansion
\vspace{-8ex}
\end{table}%
%EndExpansion

\section{Numerical results}

\begin{figure*}[th!]
\begin{center}
\hspace{-1em}\hbox{
\subfigure[Mean cell load]
{\includegraphics[
width=0.34\linewidth
]%
{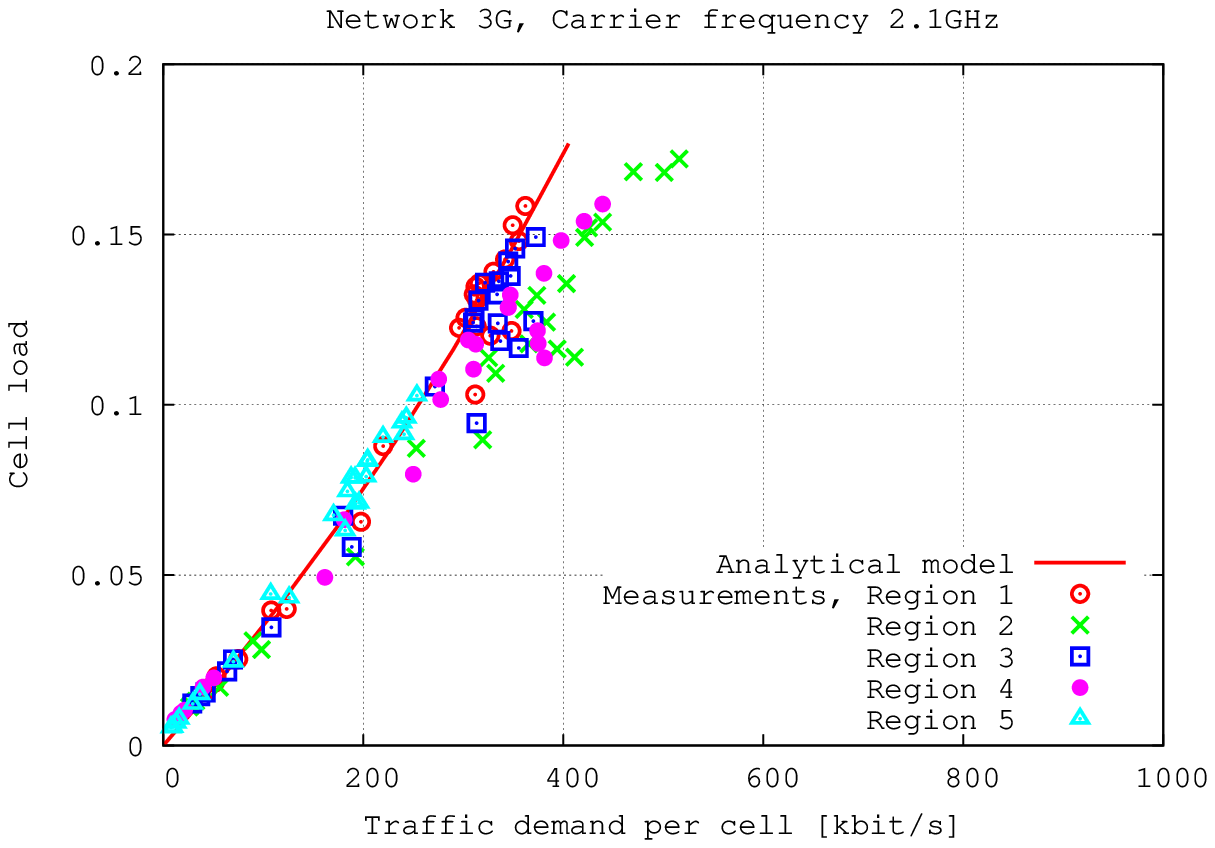}%
%\caption{Mean cell load versus traffic demand calculated analytically and
%estimated from measurements for a 3G network in Country~1.}%
\label{f.LoadVsTraffic_10712}%
}
\hspace{-1.5em}
\subfigure[Mean number of users]
{\includegraphics[
width=0.34\linewidth
]%
{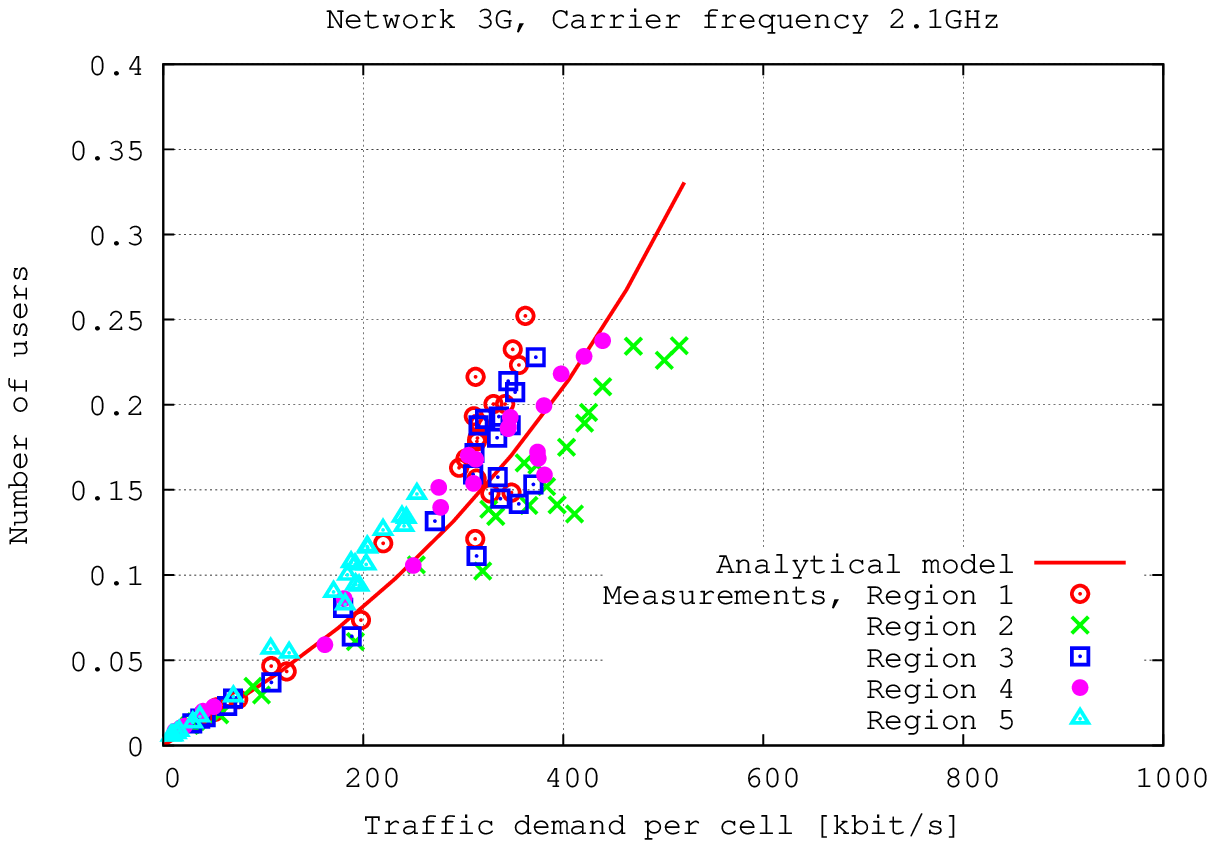}%
%\caption{Mean number of users versus traffic demand for 3G in Country~1.}%
\label{f.NbUsersVsTraffic_10712}%
}
\hspace{-1.5em}
\subfigure[Mean user throughput]
{
\includegraphics[%
width=0.34\linewidth
]%
{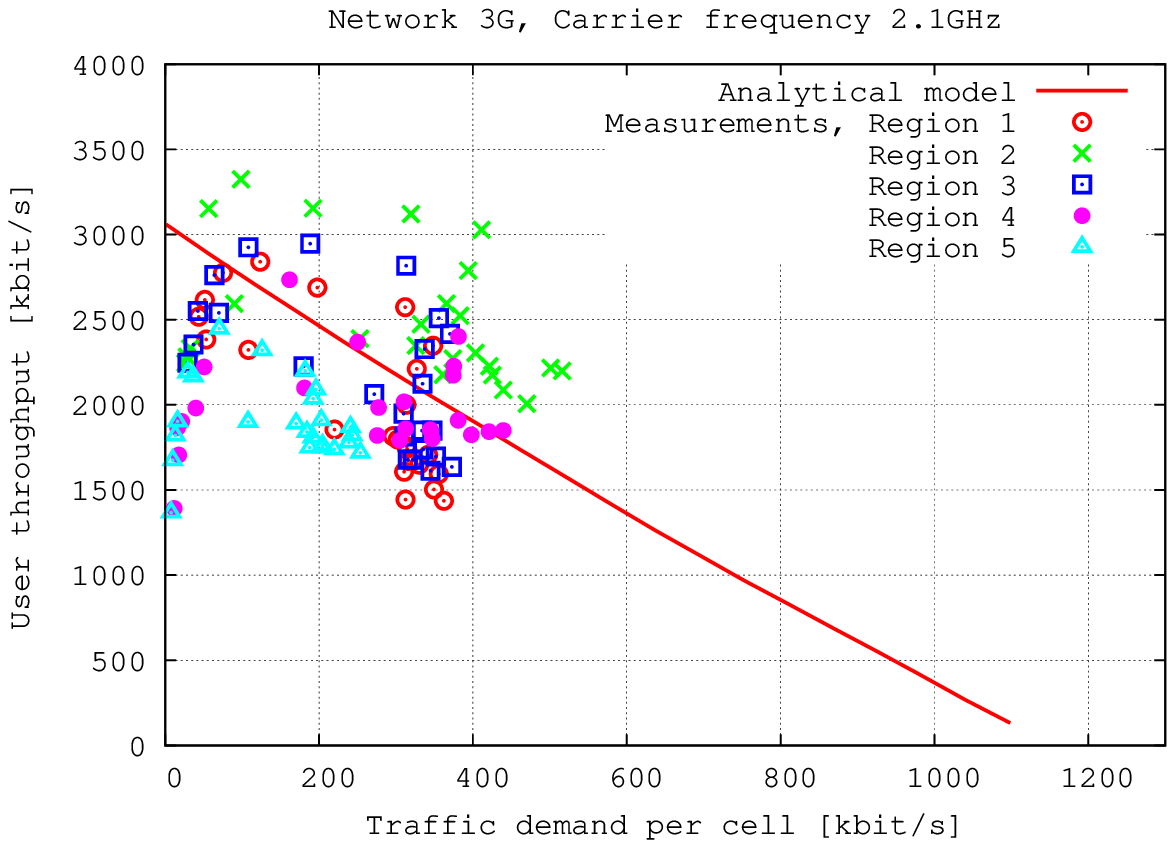}%
%\caption{Mean user's throughput in the network versus traffic demand for 3G in
%Country~1.}%
\label{f.ThroughputVsTraffic_10712}%
}
}
\end{center}
\vspace{-3ex}
\caption{Mean cell characteristics  versus traffic demand calculated analytically and
estimated from measurements for a 3G network in Country~1.
\label{Figure3G}
}%
%\end{figure*}
%\begin{figure*}[th!]
\vspace{-1ex}
\begin{center}
\subfigure[Regular decomposition into
squares of different size]
{\includegraphics[
width=0.4\linewidth
]%
{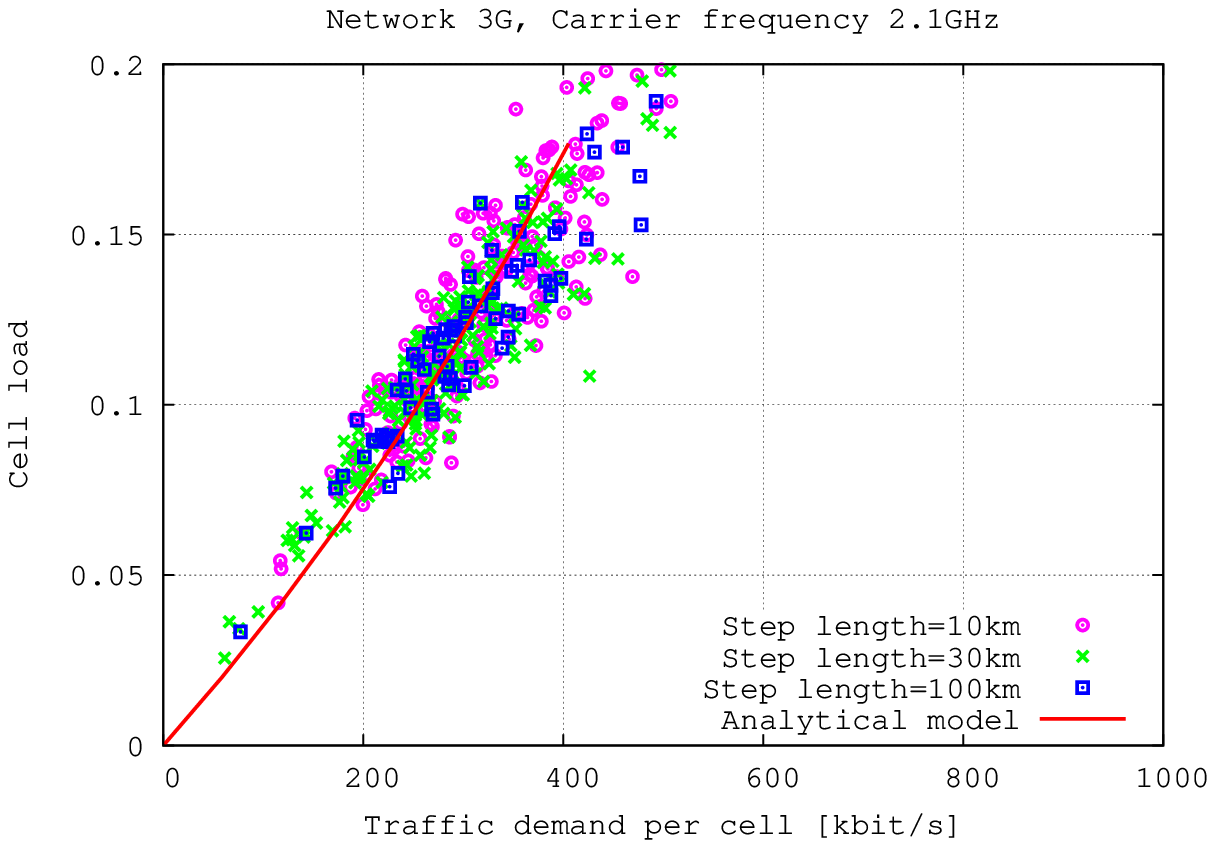}%
%\caption{Mean cell load versus traffic demand in Country~1 decomposed into
%squares of side $3,10,30$ and $100$km.}%
\label{f.LoadVsTraffic}%
}
\hspace{0.1\linewidth}
\subfigure[Decomposition into different density zones]
{\includegraphics[
width=0.4\linewidth
]%
{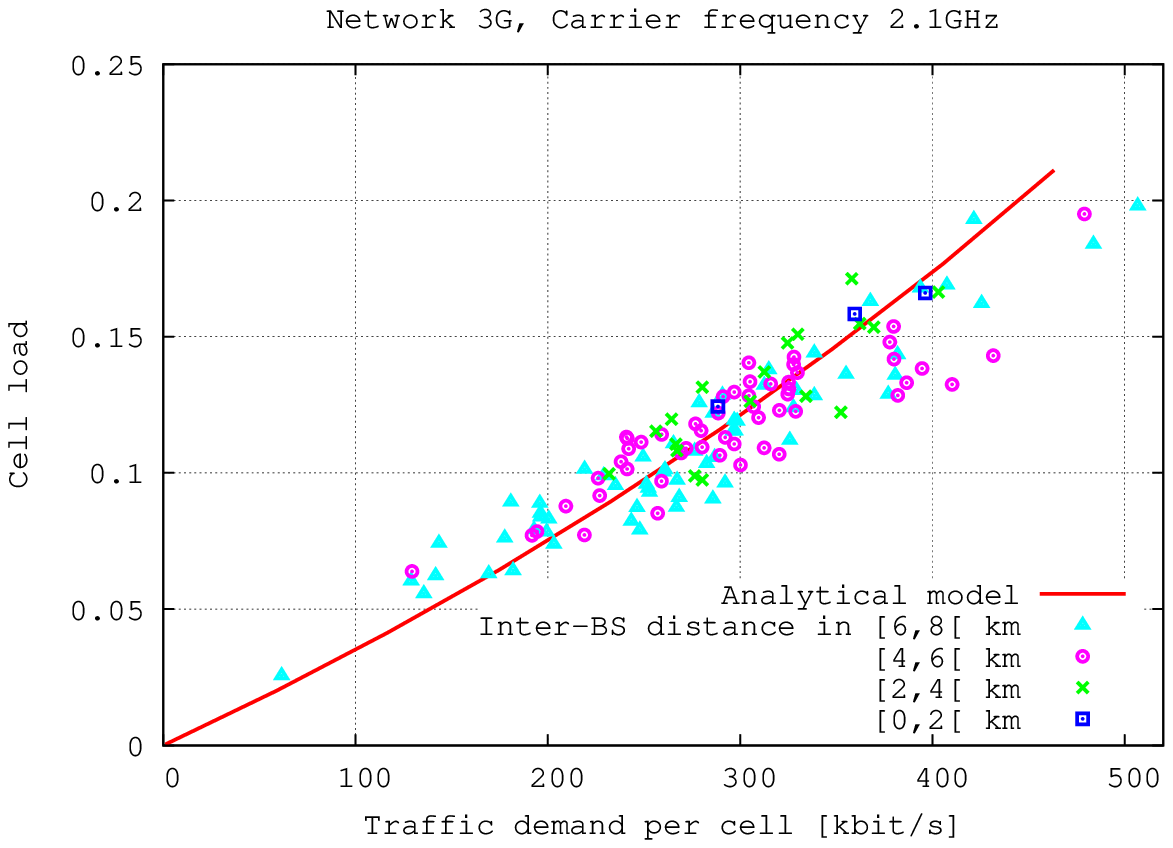}%
%\caption{Mean cell load versus traffic demand for Country~1 with mesh size
%$10$km.}%
\label{f.LoadVsTraffic_Radius}%
}
\end{center}
\vspace{-3ex}
\caption{Mean cell load versus traffic demand in a 3G network at 2.1GHz in Country~1, with two different decompositions.
\label{Figure4G}
}%
\vspace{-3ex}
\end{figure*}

\label{s.NumericalResults}The relations of the network performance metrics to
the traffic demand described in\ Section~\ref{s.MeanCell} were already
validated in~\cite{BlaszczyszynJK2014TypicalCell} comparing them to
operational network measurements in some typical cities in Europe. Thus we
shall focus on the validation of its extension to more large areas, and
ultimately to a whole country as proposed in\ Section~\ref{s.ModelAnalysis}.
Once this validation carried, we shall solve numerically the bandwidth
dimensioning problem.

\subsection{Real-field measurements}

We describe now the real-field measurements. The raw data are collected using
a specialized tool which is used by operational engineers for network
maintenance. This tool measures several parameters for every base station 24
hours a day. In particular, one can get the traffic demand, cell load and
number of users for each cell in each hour. Then we estimate the network
performance metrics for each hour averaging over all considered cells. The
mean user throughput is calculated as the ratio of the mean traffic demand to
the mean number of users. The mean traffic demand is used as the input of our
analytical model.

\subsection{A reference country}

We consider a reference country, called in what follows Country~1, divided into 5 regions. Each of these regions comprises
a large mix of urban, suburban and rural zones. In what follows we consider 3G and 4G networks deployed in Country~1.
\subsubsection{3G network}

\paragraph{Numerical setting}

\label{s.NumericalSetting}Country~1 is covered by a 3G network at carrier
frequency $f_{0}=2.1$GHz with frequency bandwidth $W=5$MHz, base station power
$P=60$dBm and noise power $N=-96$dBm.\ The peak bit-rate equals to $30$\% of
the ergodic capacity of the flat fading channel; i.e., $R\left(
\mathrm{SINR}\right)  =0.3\times W\mathbf{E}\left[  \log_{2}\left(
1+\left\vert H\right\vert ^{2}\mathrm{SINR}\right)  \right]  $\ where
$\left\vert H\right\vert ^{2}$\ is a unit mean exponential random variable
representing the fading and $\mathbf{E}\left[  \cdot\right]  $\ is the
expectation with respect to~$H$.

\paragraph{Typical urban zone}

\label{s.UrbanZone}We consider a typical \emph{urban} zone (a city) in the
considered Country~1. Knowing all BS coordinates and the surface of the
deployment zone we deduce the BS density $\lambda=1.15$ stations per km$^{2}$
which corresponds to an average distance between two neighboring BS of
$D_{\mathrm{u}}\simeq1/\sqrt{\lambda}\simeq1$km.

For the analytical model, the locations of BS is modelled by a homogeneous
Poisson point process with intensity $\lambda$. The parameters of the distance
loss model~(\ref{e.DistanceLoss}) are\ estimated from the COST
Walfisch-Ikegami model~\cite{Cost231_1999} which gives
\[
\beta=3.8,\quad K_{0}=7117\text{km}^{-1}.%
\]
The shadowing random variable is lognormal with unit mean and
logarithmic-standard deviation $\sigma=9.6$dB and mean spatial correlation
distance $50$m.

Each BS comprises three antennas having each a three-dimensional radiation
pattern specified in~\cite[Table A.2.1.1-2]{3GPP36814-900}. The BS and mobile
antenna heights equal $30$m and $1.5$m respectively. The pilot channel power
is taken equal to $10\%$ of the total base station power.

The theoretical relations of the network performance metrics to the traffic demand are
calculated for such typical urban zone and  compared, to measurements in large
areas, and ultimately the whole country, applying the approach developed  Section~\ref{s.ModelAnalysis}.

\begin{figure*}[th!]
\begin{center}
\subfigure[2.6GHz]
{\includegraphics[
width=0.4\linewidth
]%
{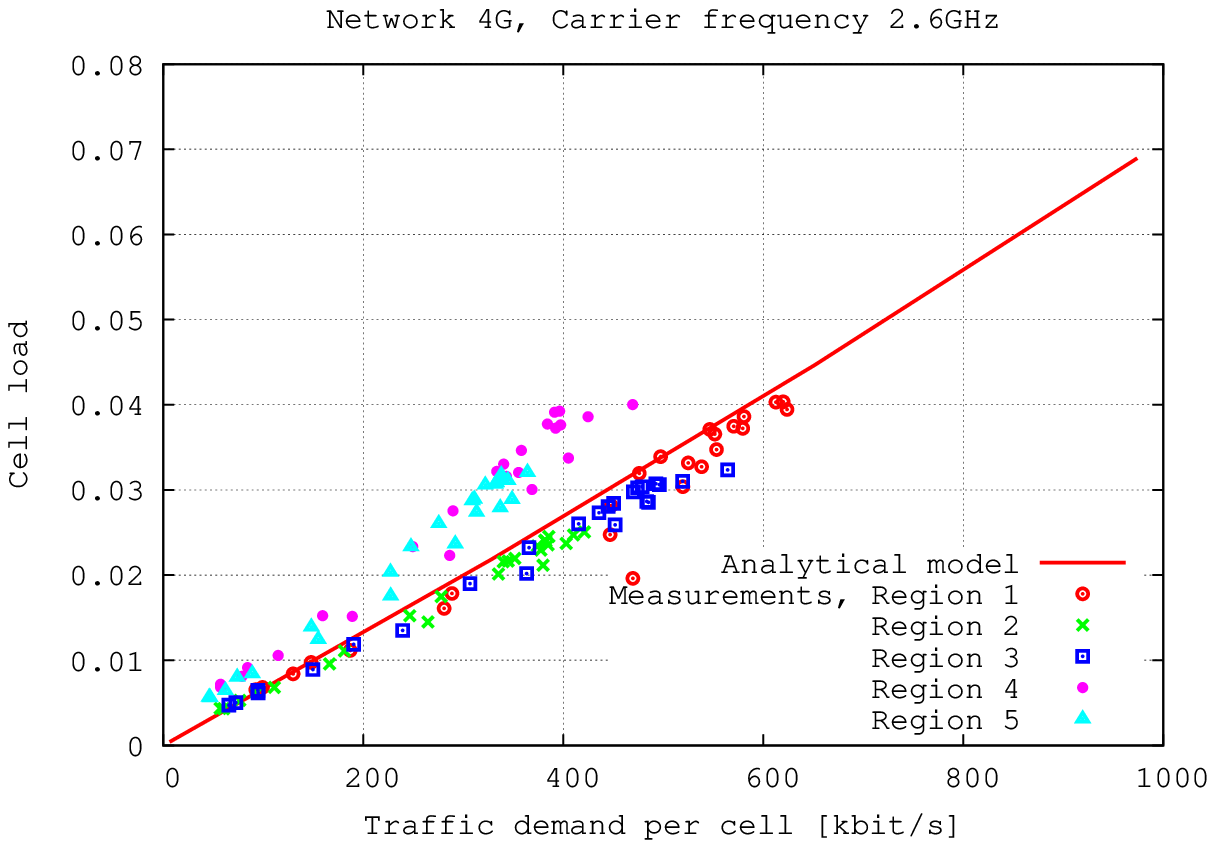}%
%\caption{Cell load versus traffic demand for 4G in Country 1.}%
\label{f.LoadVsTraffic_2600}%
}
\hspace{0.1\linewidth}
\subfigure[800MHz]
{\includegraphics[
width=0.4\linewidth
]%
{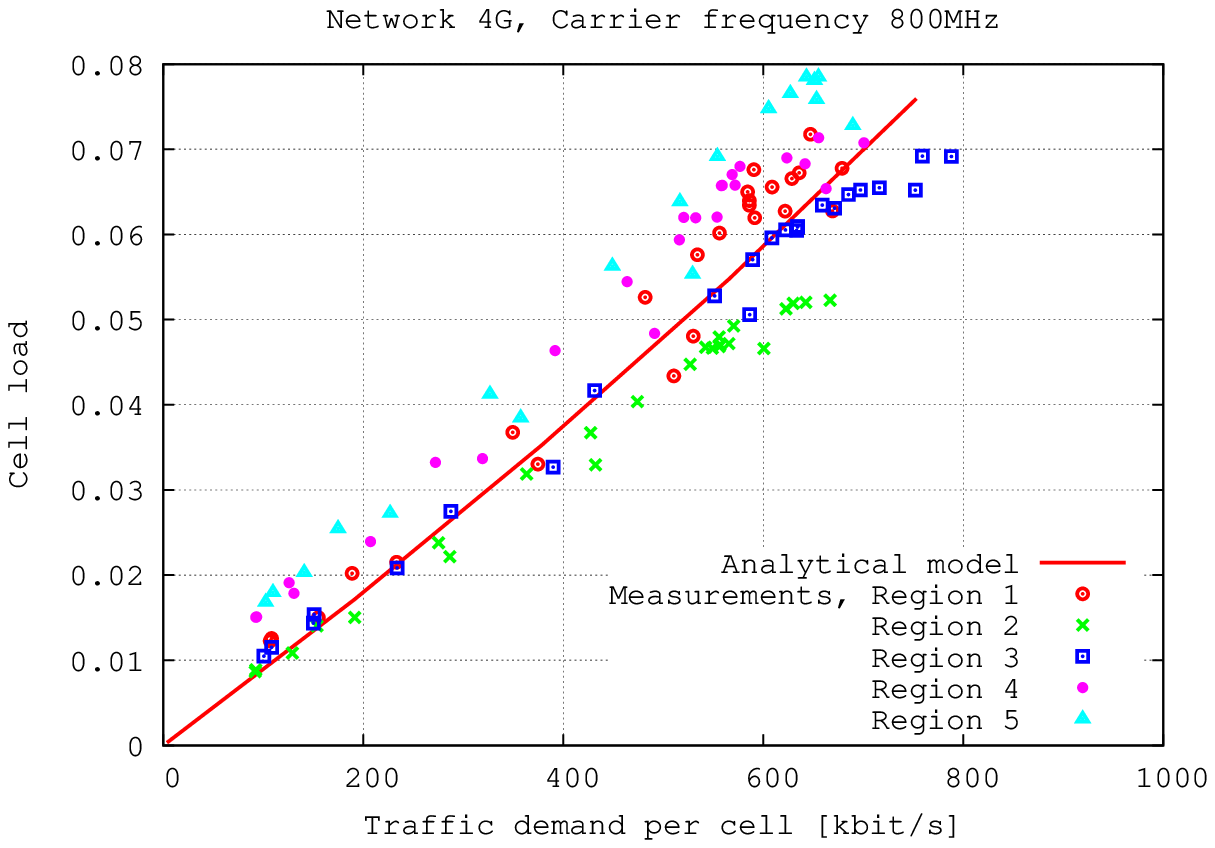}%
%\caption{Cell load versus traffic demand calculated for 4G in Country 1.}%
\label{f.LoadVsTraffic_800}%
}
\end{center}
\vspace{-3ex}
\caption{Mean cell load versus traffic demand for 4G in Country~1 at two different frequency bandwidths.
\label{Figure24G}
}%
%\end{figure*}
\vspace{-1ex}
%\begin{figure*}[th!]
\begin{center}
\subfigure[2.6GHz]
{\includegraphics[
width=0.4\linewidth
]%
{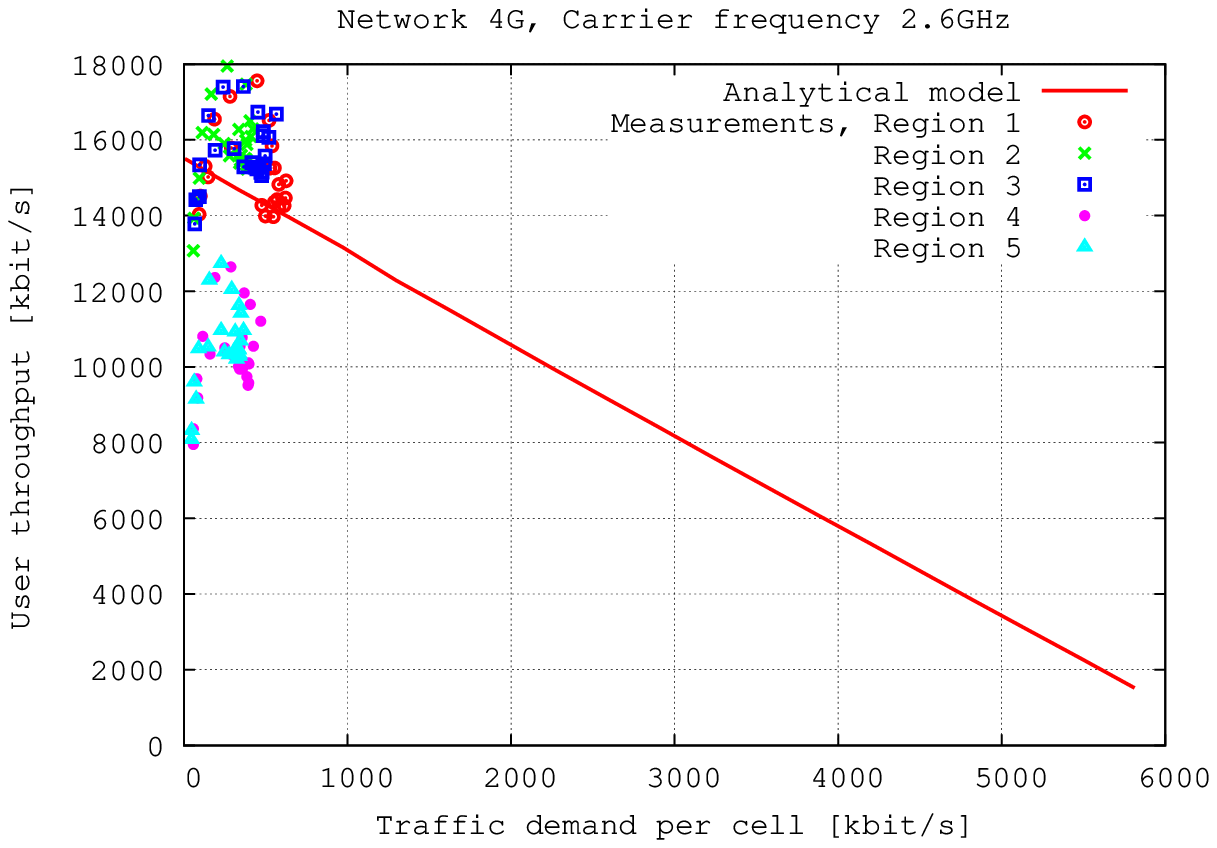}%
%\caption{User throughput versus traffic demand for 4G in Country 1.}%
\label{f.ThroughputVsTraffic_2600}%
}
\hspace{0.1\linewidth}
\subfigure[800MHz]
{\includegraphics[
width=0.4\linewidth
]%
{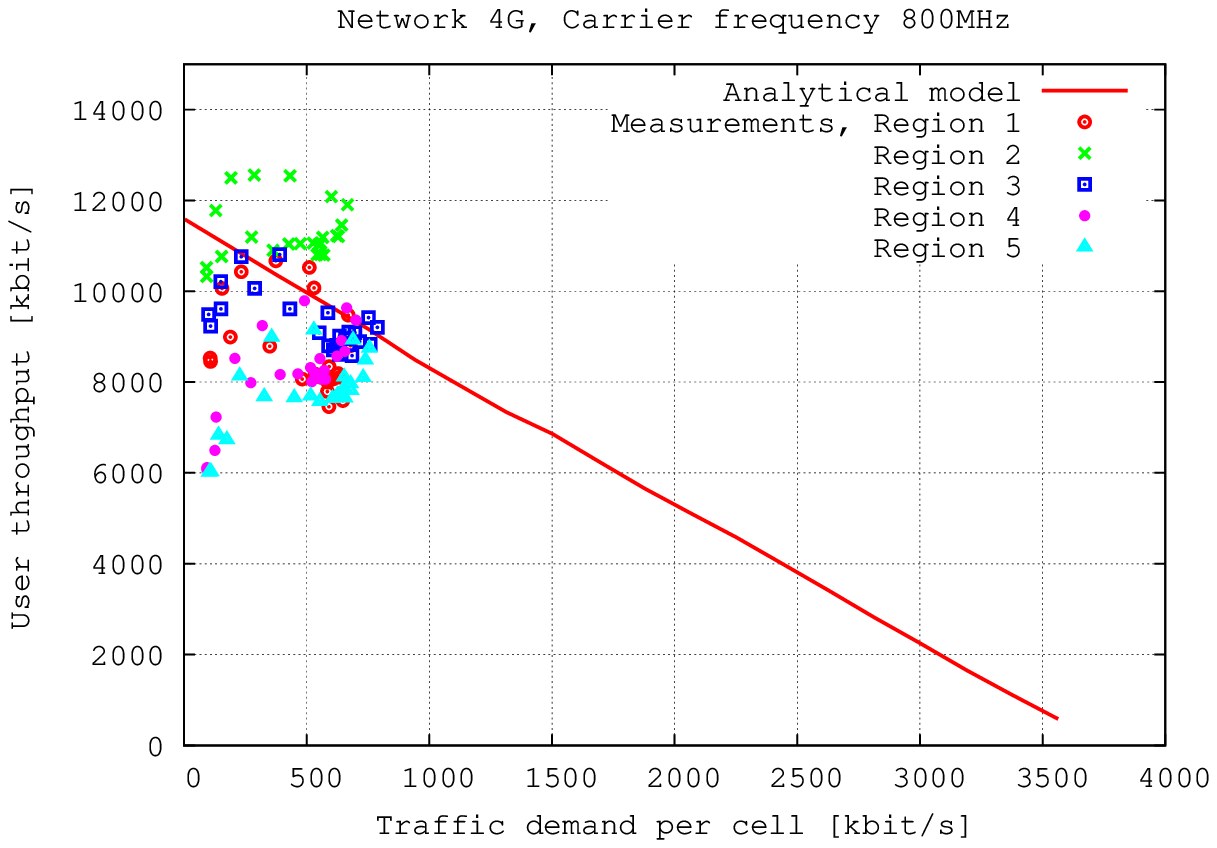}%
%\caption{User throughput versus traffic demand calculated for 4G in Country 1.}%
\label{f.ThroughputVsTraffic_800}%
}
\end{center}
\vspace{-3ex}
\caption{Mean user throughput versus traffic demand for  4G in Country~1 at two different frequency bandwidths.
\label{Figure24G-throughput}
}%
\vspace{-3ex}
\end{figure*}

\paragraph{Network performance}

Figure~\ref{f.LoadVsTraffic_10712} shows the mean cell load $\bar{\theta}$
versus mean traffic demand per cell $\bar{\rho}=\rho/\lambda$ calculated
analytically and obtained from measurements for the 5 considered regions. The
mean number of users $\bar{N}$\ and user's throughput $\bar{r}$ versus mean
traffic demand per cell $\bar{\rho}$ are plotted in
Figures~\ref{f.NbUsersVsTraffic_10712} and~\ref{f.ThroughputVsTraffic_10712}
respectively. Observe that the analytical curves fits well to the measurements.%

In order to analyze the effect of the size of the regions over which we
average the cell performance metrics, we decompose Country~1 into a regular
grid composed of squares (meshes) of equal side. Figure~\ref{f.LoadVsTraffic}
shows the mean cell load versus traffic demand for mesh sizes $3,10,30$ and
$100$km respectively. Observe that the larger is the mesh size, the closer are
measurements to the analytical (mean) model.%

We aim now to check the hypothesis made in Section~\ref{ss.KversusDelta}, namely that
zones with different density of base stations observe the same dependence of the mean 
cell characteristics on the mean per-cell traffic demand.
Figure~\ref{f.LoadVsTraffic_Radius} shows the mean cell load versus traffic
demand for mesh size $30$km where we distinguish the meshes according to the
distance between neighboring BS; distances within $[0,2]$~km\ (resp.
$[6,8]$~km) corresonding to urban (resp. rural) zones. Observe that there is
no apparent difference between urban, sub-urban and rural meshes.%

\subsubsection{4G network}

The peak bit-rate is calculated by $R\left(  \mathrm{SINR}\right)  =b\times
W\log_{2}\left(  1+\mathrm{SINR}/a\right)  $ where $a=3$, $b=1.12$ are
estimated from fitting to the results of a link simulation tool.\ 

\paragraph{Numerical setting for carrier frequency $2.6$GHz}

We consider a 4G network at carrier frequency $f=2.6$GHz with frequency
bandwidth $W=20$MHz, base station power $P=63$dBm and noise power $N=-90$dBm.\ 

The typical urban zone is the same as described in Section~\ref{s.UrbanZone}
except for the distance loss parameter $K$. Following the COST-Hata
model~\cite{Cost231_1999} this latter is calculated by
assuming that it is proportional to $f^{2/\beta}$; which gives $K=K_{0}%
\times\left(  f/f_{0}\right)  ^{2/\beta}=7964$km$^{-1}$.

\paragraph{Numerical setting for carrier frequency $800$MHz}

We consider also a 4G network at carrier frequency $f=800$MHz with frequency
bandwidth $W=10$MHz, base station power $P=60$dBm and noise power $N=-93$dBm.\ 

The typical urban zone is the same as described in Section~\ref{s.UrbanZone}
except for the average distance between two neighboring BS which is now
$D_{\mathrm{u}}\simeq1.5$km and the distance loss parameter $K=K_{0}%
\times\left(  f/f_{0}\right)  ^{2/\beta}=4283$km$^{-1}$.

\paragraph{Network performance}

Figures~\ref{f.LoadVsTraffic_2600} and~\ref{f.LoadVsTraffic_800} show the mean
cell load versus mean traffic demand per cell calculated analytically and
obtained from measurements for the 4G network at carrier frequencies $2.6$GHz
and $800$MHz\ respectively. Observe again that the analytical curve fits well
with the measurements.%

Figures~\ref{f.ThroughputVsTraffic_2600} and~\ref{f.ThroughputVsTraffic_800}
show the mean user's throughput versus mean traffic demand per cell for the 4G
network at carrier frequencies $2.6$GHz and $800$MHz\ respectively.
Even if currently we do not observe traffic demand larger than
1000kbit/s per cell, we show the predicted curves in their whole domains 
as they will serve for the network dimensioning, cf. Section~\ref{ss.dimension}.

\begin{figure*}[th!]
\begin{center}
\hspace{-1em}\hbox{
\subfigure[Mean cell load]
{\includegraphics[
width=0.34\linewidth
]%
{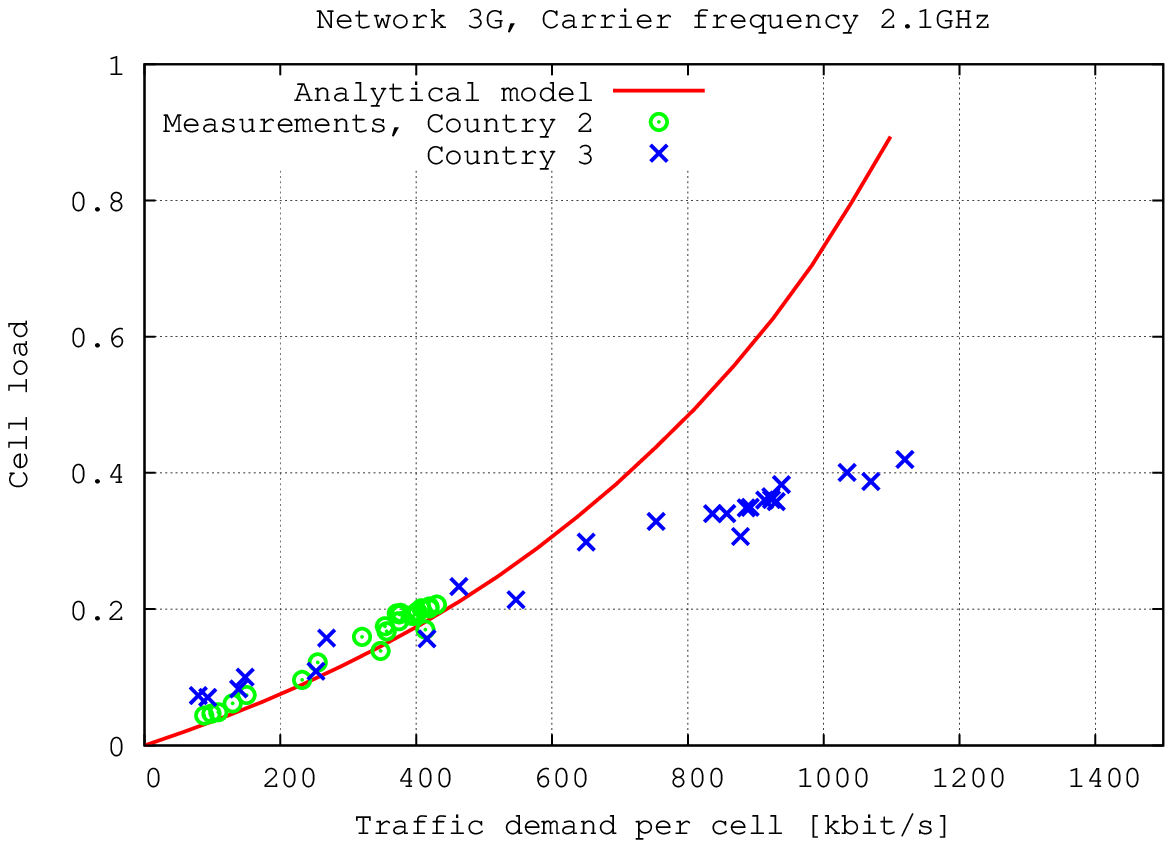}%
\label{f.LoadVsTraffic_Countries}%
}
\hspace{-1.5em}
\subfigure[Mean number of users]
{\includegraphics[
width=0.34\linewidth
]%
{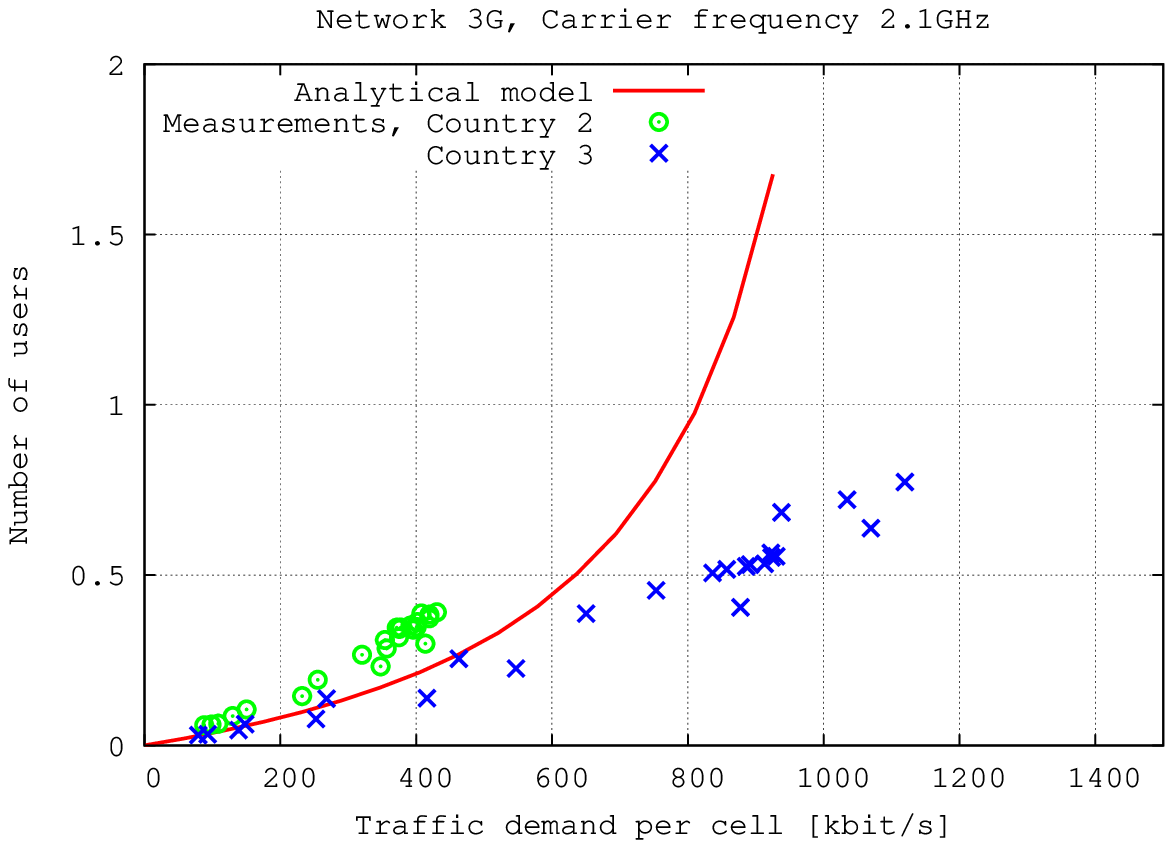}%
\label{f.NbUsersVsTraffic_Countries}%
}
\hspace{-1.5em}
\subfigure[Mean user throughput]
{
\includegraphics[%
width=0.34\linewidth
]%
{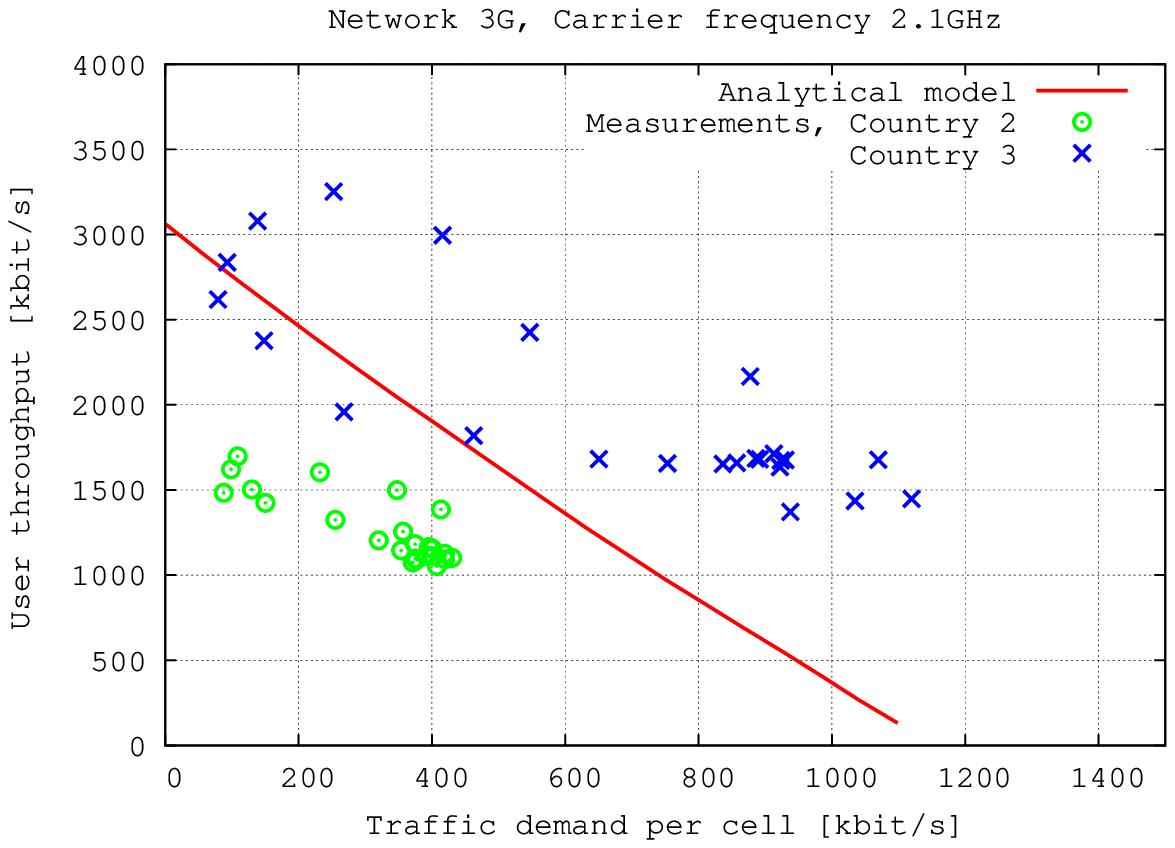}%
\label{f.ThroughputVsTraffic_Countries}%
}
}
\end{center}
\vspace{-3ex}
\caption{Mean cell characteristics  versus traffic demand calculated analytically and
estimated from measurements for a 3G network in Country~2 and~3.
\label{Figure3G-Coutry2and3}
}%
\vspace{-3ex}
\end{figure*}

\subsection{Two other countries}

We consider now two 3G networks in two other countries, called Country~2 and Country~3. 
The numerical setting is the
same as in Section~\ref{s.NumericalSetting}. Since we don't have detailed
information on the propagation characteristics in these countries, we make the
hypothesis that the typical urban zone has the same characteristics as for Country~1; i.e. those described in Section~\ref{s.UrbanZone}.

Figures~\ref{f.LoadVsTraffic_Countries},~\ref{f.NbUsersVsTraffic_Countries}
and~\ref{f.ThroughputVsTraffic_Countries} show the mean cell load, users
number and user's throughput versus mean traffic demand per cell for the two
considered countries. Note that for Country~3 the user's throughput
stagnates for traffic demand exceeding $800$kbit/s/cell. This is due to
congestion; some users being blocked at access as may be seen in
Figure~\ref{f.NbUsersVsTraffic_Countries}.%

\subsection{Bandwidth dimensioning}
\label{ss.dimension}
We come back now to the question stated in the title of this article:
What frequency bandwidth is required  to run cellular network
guaranteeing sufficient QoS for its users? This is the frequency
dimension problem often faced by the network operators.

Here we consider the mean user throughput in the downlink as the QoS metric.
We have remarked at the end of Section~\ref{ss.KversusDelta}
that (conjecturing) the relation~(\ref{e.fractal}) allows one to 
capture the key dependence between the mean user throughput and the
per cell traffic demand (given the network density and the frequency bandwidth)
for an inhomogeneous  network focusing only on one homogeneous type of network
area. We use for this purpose the urban area. 
Examples of the corresponding user-throughput versus traffic demand
curves have been presented on Figures~\ref{f.ThroughputVsTraffic_10712}, 
\ref{f.ThroughputVsTraffic_2600} and~\ref{f.ThroughputVsTraffic_800}.
Note that these curves depend on (increase with) the frequency bandwidth.
The solution of the frequency dimensioning problem (for a given
technology  and a reference, urban, inter-BS distance) consists in finding, 
for each  value of the traffic demand per cell, the minimal  
frequency bandwidth such that the predicted mean user throughput reaches
a given target value.

Figure~\ref{f.BandwidthVsTraffic}
plots the solutions for this problem in several considered cases.
More precisely, it gives the predicted frequency
bandwidth allowing one to assure the target $\bar{r} =5$ Mbit/s mean
user's throughput in the network,
for a given (per cell)  traffic
demand. We show solutions  for 3G and 4G networks with different frequency carriers $f$ and
different inter-BS distances $D_{\mathrm{u}}$ relative  to  the urban area.

As a possible application of the above dimensioning strategy let us
mention the following actual problem.
In Country~1, the configurations currently deployed for 4G are $f=2.6$GHz\ with $D_{\mathrm{u}%
}=1$km and $f=800$MHz\ with $D_{\mathrm{u}}=1.5$km. The operator is
wondering whether to densify the network for $f=800$MHz\ decreasing
$D_{\mathrm{u}}$ to $1$km. Figure~\ref{f.BandwidthVsTraffic}\ shows
how much bandwidth can be economized in this case.%

\section{Conclusion}

We propose a model permitting to calculate the performance of 3G and 4G
wireless cellular networks at the scale of a whole country comprising urban,
suburban and rural zones. This model relies on an observation that the
distance coefficient in the propagation loss function depends on the type of
zone in such a way that its product to the distance between neighboring base
stations remains approximately constant. It is then shown that the network
performance, in terms of the relations of the mean cell load, number
of users and user's throughput to mean traffic demand, are the same for the different types
of zones.

This theoretical result is validated by field measurements in 3G and 4G
cellular networks in various countries. Then we solve the bandwidth
dimensioning problem for 3G and 4G networks at different carrier frequencies;
i.e. we plot the frequency bandwidth as function of the traffic demand per
cell to assure a user's throughput of $5$Mbit/s. This curve is crucial for
operators to predict the frequency bandwidth required to serve the continuing
increase of traffic in the next decades.

We shall attempt in future work to extend the present approach to the uplink;
in particular to account for power control effect and the specificity of the
locations of the interferers in this case.

\begin{acknowledgement}
We thank our colleagues from Orange for fruitful exchanges related to the
present work.
\end{acknowledgement}

\begin{figure}
[t!]
\begin{center}
\includegraphics[
width=0.8\linewidth
]%
{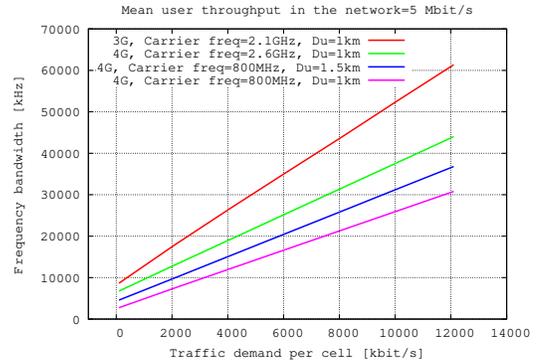}%
\caption{Bandwidth versus traffic demand per cell to assure a mean user's
throughput in the network $\bar{r}=5$Mbit/s.}%
\label{f.BandwidthVsTraffic}%
\end{center}
\vspace{-3ex}
\end{figure}

\addtocounter{section}{1}
\addcontentsline{toc}{section}{References} 
{\small
\bibliographystyle{IEEEtran}
\bibliography{Cellular,Probability}
}
\end{document}